%% file: usenix.tex
\documentclass[letterpaper,twocolumn,10pt]{article}

\usepackage{usenix}

\usepackage[english]{babel}
\usepackage{booktabs}       
\usepackage{multirow}
\usepackage{enumitem}       
\usepackage{caption}
\usepackage{graphicx}
\usepackage{xspace}
\usepackage{amsmath}
\usepackage{algorithm}
\usepackage{algpseudocode}
\usepackage{listings}
\usepackage{xcolor}
\usepackage{float}
\usepackage{pifont}  
\usepackage{amsthm}
\usepackage{amssymb}

\newif\ifdraft
\drafttrue

\newcommand{\sys}{LiveR\xspace}

\newcommand{\cmark}{\ding{51}}  
\newcommand{\xmark}{\ding{55}}  
\newtheorem{theorem}{Theorem}
\newtheorem{definition}{Definition}[section]

\setlist[itemize]{leftmargin=*, noitemsep, topsep=0pt}
\setlist[enumerate]{leftmargin=*, noitemsep, topsep=0pt}

\newcommand{\para}[1]{\paragraph{#1}}

\begin{document}

\date{}

\title{\Large \bf \sys: Fine-Grained Elasticity via Live Reconfiguration for Model Training}

\author{
{\rm Haoyuan Liu}\\
Shanghai Jiao Tong University
\and
{\rm Kairui Zhou}\\
Shanghai Jiao Tong University
\and
{\rm Shuyao Qi}\\
Shanghai Jiao Tong University
\and
{\rm Qinwei Yang}\\
Shanghai Jiao Tong University
\and
{\rm Shengkai Lin}\\
Shanghai Jiao Tong University
\and
{\rm Shizhen Zhao}\\
Shanghai Jiao Tong University
\and
{\rm Wei Zhang}\\
University of Connecticut
}

\maketitle

\thispagestyle{empty}

\begin{abstract}
\input{abstract}
\end{abstract}

\input{intro}
\input{background}
\input{motivation}
\input{design}

\input{implementation}
\input{evaluation}

\input{conclusion}

{\footnotesize 
\bibliographystyle{plain}
\bibliography{sample}
}

\newpage
\appendix
\input{appendix}

\end{document}

%% file: abstract.tex
To reduce user costs and maximize cluster utilization, large model training increasingly leverages volatile but inexpensive GPU capacity, such as spot instances and reclaimable resources in shared clusters. Yet, capitalizing on these economic benefits requires jobs to adapt within the short warning windows that many such environments provide. Existing elastic training systems still treat reconfiguration as stop-and-restart: they externalize distributed state through checkpoints, rebuild the distributed runtime on a new topology, and restart training, turning each resize event into a storage-heavy recovery procedure that incurs substantial downtime from checkpoint I/O, process restart, CUDA initialization, and communicator setup. We present \sys, a live reconfiguration runtime for elastic LLM training that replaces storage-backed restart with a live, bounded-memory handoff between mixed-parallel training worlds. While the current world continues training, \sys asynchronously prepares the target world, bootstraps newly added workers in isolation to keep heavyweight initialization off the critical path, and streams model state directly over high-bandwidth interconnects while reshaping it online across tensor, pipeline, and data parallel dimensions. Once the target world is ready, \sys performs a lightweight commit that switches training to the new configuration without stop-and-restart on the live path. We implement \sys atop Megatron-LM and PyTorch and evaluate it end-to-end on a multi-node GPU cluster. Across diverse reconfiguration scenarios, \sys reduces downtime from minutes to seconds, accelerates reconfiguration by 14$\times$-23$\times$ over checkpoint/restart baselines, incurs minimal steady-state overhead, and sustains up to 99\% training goodput under volatile-resource conditions, making volatile low-cost GPU capacity far more practical for LLM training.

%% file: intro.tex
\section{Introduction}
\label{sec:intro}

Training large language models (LLMs) at frontier scale is extraordinarily expensive, often requiring millions of dollars and months of continuous GPU execution~\cite{openai2024gpt4technicalreport, deepseekai2025deepseekv3technicalreport, scaling-laws, grattafiori2024llama3herdmodels, cottier2025risingcoststrainingfrontier, workshop2023bloom176bparameteropenaccessmultilingual}. A natural way to reduce this cost is to train on cheaper capacity, such as spot instances~\cite{aws-spot, spotnik} and shared clusters whose idle GPUs can be harvested and later reclaimed, even though such resources are often less reliable (Figure~\ref{fig:spot-preemption}). Cheap capacity is useful only if training can reconfigure quickly enough to follow resource changes; otherwise, the savings are eroded by downtime at each resize. To benefit from volatile but inexpensive capacity, LLM training must therefore be \emph{elastic}: it must scale out when GPUs become available and scale in when they are reclaimed, without paying minutes of lost progress at every transition.

Yet support for elastic LLM training remains limited, even though spot-like resources are useful only when training can adapt to them. In practice, when resources change, training is often adapted through restart-based workflows: the job checkpoints or otherwise terminated directly, relaunches under a new resource configuration, and reconstructs the runtime on the new topology. This approach is compatible with existing frameworks such as Megatron-LM~\cite{megatron-lm} and DeepSpeed~\cite{deepspeed}, and can accommodate a wide range of resource changes, but it incurs substantial downtime from checkpoint I/O, process restart, CUDA initialization, and communicator construction, even with faster checkpointing~\cite{bcp, fast-checkpoint, universal, datastatesllm}. 
A smaller set of online approaches attempts to support elasticity without full restarts (e.g., Bamboo~\cite{bamboo}, Oobleck~\cite{oobleck}), but typically does so by restricting which dimensions of parallelism may change during a resize or by retaining model-state ownership assumptions tied to a fixed layout. As a result, these systems can react quickly yet still preserve imbalanced pipeline partitions, inefficient tensor layouts, or communication patterns that are poorly matched to the new resources. In practice, current approaches often provide either quick resizing or a good post-reconfiguration layout, but not both.

When resource changes are known ahead of time, this tradeoff is not fundamental; it is largely created by how current systems handle runtime setup and distributed state. Restart-based workflows place reconfiguration on the critical path by routing state through storage and reconstructing the distributed runtime from scratch. Online systems~\cite{zhang2024rubickexploitingjobreconfigurability, gandhi2024recycleresilienttraininglarge, he2021pipetransformerautomatedelasticpipelining}, in contrast, often entangle model-state ownership with a fixed parallelism structure, making broad reshaping difficult. This suggests a different abstraction: reconfiguration should be a \emph{handoff}, not a \emph{restart}. The current training world should keep making progress while the target world is prepared in the background, after which model state should be transferred directly over high-bandwidth interconnects and reshaped across TP, PP, and DP dimensions. Realizing such a handoff is non-trivial. Model states, and pipeline-stage ownership must be remapped consistently across mixed-parallel layouts; the new communication topology must be constructed without stalling training; newly added workers must complete heavyweight local initialization off the critical path; and the handoff must preserve bounded memory usage.

We present \sys, a live reconfiguration runtime for elastic LLM training under pre-announced resource changes. The key idea in \sys is a \emph{live, bounded-memory handoff} between mixed-parallel training worlds. Rather than treating reconfiguration as a storage-heavy stop-and-restart procedure, \sys keeps the source world running, prepares the destination world asynchronously, and streams model state directly between them while reshaping it across a broad range of TP/PP/DP configurations. This design removes checkpoint I/O, process restart, and distributed runtime rebuild from the reconfiguration critical path, while decoupling target-world setup from training progress. \sys is designed for reconfiguration events that are announced before they take effect, including spot revocations, scheduler-driven reallocations in shared GPU clusters, and planned resizes. For unexpected fail-stop events, where GPU state becomes unreachable, \sys falls back to conventional checkpoint-based recovery rather than attempting to replace general fault tolerance.

We implement \sys atop Megatron-LM and PyTorch, and evaluate it end-to-end on a multi-node GPU cluster. Across our experiments, \sys reduces reconfiguration downtime from minutes to seconds, achieves a 14$\times$--23$\times$ speedup over checkpoint/restart baselines, incurs minimal steady-state overhead, and sustains up to 99\% goodput under volatile-resource scenarios.

We make the following contributions:
\begin{itemize}
    \item We design \sys, a live reconfiguration architecture for elastic LLM training under pre-announced resource changes that overlaps target-world setup with ongoing training and removes checkpoint I/O and distributed runtime rebuild from the reconfiguration critical path.
    \item We develop a bounded-memory mixed-parallel state handoff mechanism that consistently remaps model states and pipeline-stage ownership across a broad range of TP/PP/DP reconfigurations.
    \item We implement \sys atop Megatron-LM and PyTorch, and show through end-to-end experiments on a multi-node GPU cluster that live reconfiguration substantially reduces downtime while preserving high training efficiency under resource volatility.
\end{itemize}

%% file: background.tex
\section{Background and Related Work}
\label{sec:background}

\subsection{Volatile Resources Make Elasticity Necessary}
\label{sec:bg-volatility}

LLM training runs for weeks on large GPU fleets, making infrastructure efficiency a first-order concern.
In practice, many clusters still exhibit substantial idle capacity due to scheduler fragmentation and over-provisioned static reservations~\cite{tiresias, gandiva, antman, pollux, sia, themis}.
At the same time, operators increasingly rely on discounted spot/preemptible instances~\cite{aws-spot, skypilot}, where node availability is volatile and preemption notices are short ($\sim$2 minutes~\cite{aws-spot}).
These two forces create a clear need for elastic training: jobs should scale out when capacity appears and scale in quickly when resources are reclaimed.

\begin{figure}[t]
    \centering
    \includegraphics[width=\linewidth]{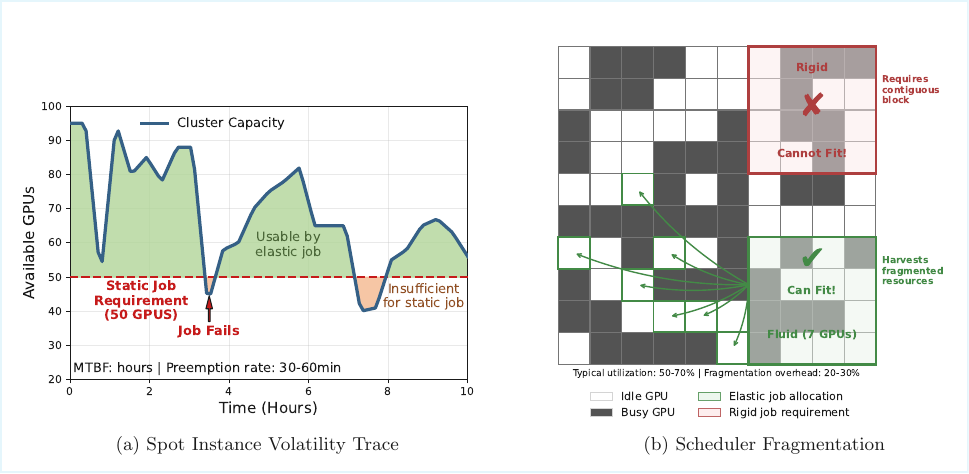}
    \caption{Spot instance preemption example. Node availability fluctuates on timescales of minutes to hours, with preemption warnings (shaded regions) providing only a brief window for state migration. .}
    \label{fig:spot-preemption}
\end{figure}
\subsection{Why Existing Elasticity Is Still Slow}
\label{sec:bg-challenges}

Although elasticity is desirable, practical LLM reconfiguration~\cite{kang2025elaswaveelasticnativescalablehybridparallel, Wagenl_nder_2024} remains expensive for two reasons: restart overhead and mixed-parallel state~\cite{liu2023ringattentionblockwisetransformers, megatron-lm, fang2024uspunifiedsequenceparallelism} complexity.

\subsubsection{Overhead of Stop-and-Restart}
\label{sec:bg-overhead}

Most production flows still use Stop-and-Restart: pause training, relaunch with a new world size, and warm up the runtime again.
This incurs minutes of downtime and often erases the cost advantage of volatile resources.
The delay comes from two components:
\begin{itemize}
    \item \textbf{Storage I/O:} Reloading large distributed states is expensive even with optimized checkpointing speed~\cite{gemini, bcp}.
    \item \textbf{Distributed Re-initialization:} Relaunching processes and rebuilding communicator topology (e.g., NCCL) and warmup is also expensive; initializing and warmup 32 GPUs with 14B model can take nearly 60 seconds even on optimized InfiniBand networks.
\end{itemize}

The key point is that faster checkpointing only addresses the first term; distributed initialization overhead remains on the critical path.

\begin{table}[t]
\centering
\caption{Breakdown of Restart latency for GPT-20B model. With 4 nodes (32 GPUs), DP=2, TP=4, PP=4.}
\label{tab:reconfig-breakdown-motivation}
\small
\begin{tabular}{lrr}
\toprule
\textbf{Phase} & \textbf{Latency (s)} & \textbf{\% of Total} \\
\midrule
Checkpoint Load        & 54.6& 42.9\%\\
Distributed Init + Warmup     & 70.1& 55.1\%\\
Misc (setup, sync, etc.)      & 2.4  & 2.0\%\\
\midrule
\textbf{Total}             & \textbf{127.1}& 100\% \\
\bottomrule
\end{tabular}
\end{table}

\subsubsection{Complexity of Mixed-Parallel Reconfiguration}
\label{sec:bg-parallelism}

Beyond raw downtime, LLM state is sharded across multidimensional parallelism (TP/PP/DP and parameter partitions)~\cite{megatron-lm, deepspeed, gspmd, alpa, galvatron}.
When GPU count changes, the old shard layout typically no longer matches the new topology.
As a result, the system must remap and move large model states while preserving correctness and bounded memory.
This is exactly where many online approaches become restricted to fixed templates (e.g., adjusting only replica count or homogeneous pipeline widths), because full mixed-parallel reshaping (arbitrary changes to TP/PP/DP degrees) is much harder than plain restart.

The above problems—volatile resources demanding elasticity, the prohibitive cost of stop-and-restart, and the complexity of resharding mixed-parallel states—have motivated a growing body of work on elastic training systems. However, existing approaches differ fundamentally in how they handle reconfiguration: some optimize checkpointing to reduce restart latency, others enable limited forms of online elasticity within fixed parallelism templates, and a few pursue fast live migration but sacrifice scalability.

\subsection{Related Work}
\label{sec:bg-prior}

Prior systems can be viewed through two baseline classes used in our evaluation, plus a hybrid category we identify as \textit{Online with Limited Reshaping}.
These categories reflect a fundamental trade-off: systems that restart can reshape parallelism freely (at the cost of minutes of downtime), while systems that avoid restarts can do so only by constraining how parallelism may change.
The root cause is that live reshaping of mixed-parallel state is far more complex than adjusting replica counts---it requires computing precise tensor-shard intersections across different TP/PP/DP decompositions and moving bytes accordingly, all without blocking the active training loop.

\textbf{(A) Restart-based method.}
Systems such as Varuna~\cite{varuna}, TorchElastic~\cite{torchelastic}, and DeepSpeed elasticity~\cite{deepspeed} retain restart-centric workflows.
Checkpointing variants including UCP~\cite{universal}, ByteCheckpoint~\cite{bcp}, and Gemini~\cite{gemini} improve storage behavior but still depend on process re-initialization.
UCP and ByteCheckpoint support \textit{reshaping} (converting between different parallelism strategies during reload), but remain restart-based.
Thus, they reduce only part of the downtime in Table~\ref{tab:reconfig-breakdown-motivation}, which was measured on our 4-node A800 testbed using a GPT-14B model with DP=2, TP=4, PP=2.

\textbf{(B) Online fixed-template method.}
Online systems such as Bamboo~\cite{bamboo}, Oobleck~\cite{oobleck}, Parcae~\cite{parcae}, and TrainMover~\cite{trainmover} avoid full restart in certain scenarios, but constrain how parallelism can change at runtime.
Bamboo uses redundant computation to tolerate failures without restart, but adjusts only pipeline depth/width within fixed templates.
Oobleck employs pipeline templates for dynamic reconfiguration without checkpointing, yet restricts reshaping to homogeneous or heterogeneous pipeline combinations that fit predefined templates.
Parcae and TrainMover provide live migration with minimal overhead, but do not support changing TP/PP/DP degrees arbitrarily.
These systems achieve Init-Free and Storage-Free operation, but sacrifice flexibility in parallelism reshaping.

\textbf{(C) System-level and scheduling work.}
System-level migration (e.g., Singularity~\cite{singularity}, CRIU-based GPU checkpointing~\cite{criugpu}) and cluster scheduling (e.g., Tiresias~\cite{tiresias}, Gandiva~\cite{gandiva}, Pollux~\cite{pollux}) are complementary.
They improve infrastructure operations but do not directly solve live mixed-parallel state reshaping in framework-level distributed training.

\textbf{(D) Automated parallelism search.}
Systems such as Alpa~\cite{alpa} and Megatron-LM's automatic TP/PP search~\cite{megatron-lm} address a complementary problem: given a static GPU topology, determine the optimal parallelism decomposition. These systems perform a one-time search at job launch and assume the topology remains fixed throughout the run. When the topology changes, they must re-run the search and restart the job. \sys solves the \emph{execution} problem---how to transition from one parallelism configuration to another without stopping---and is orthogonal to the \emph{search} problem of which configuration to choose. A natural integration would have the search system determine the target $(TP', PP', DP')$ and \sys execute the live transition.

\begin{table}[h]
\caption{Comparison of elastic training systems. \textbf{Reshaping}: supports arbitrary reconfiguration across TP/PP/DP strategies (not just fixed templates). \textbf{Storage-Free}: no persistent storage I/O on the critical path (in-memory OK). \textbf{Init-Free}: no process re-initialization or NCCL rebuild required.}
\label{tab:comparison}
\resizebox{\columnwidth}{!}{%
\begin{tabular}{@{}lcccc@{}}
\toprule
\textbf{Works} & \textbf{Mechanism} & \textbf{Reshaping} & \textbf{Storage-Free} & \textbf{Init-Free} \\ \midrule
Varuna~\cite{varuna} & Cold Restart & \xmark & \xmark & \xmark \\
UCP~\cite{universal} & Cold Restart & \cmark & \xmark & \xmark \\
ByteCheckpoint~\cite{bcp} & Cold Restart & \cmark & \xmark & \xmark \\
Bamboo~\cite{bamboo} & Online (Template) & \xmark$^*$ & \cmark & \cmark \\
Oobleck~\cite{oobleck} & Online (Template) & \xmark$^*$ & \cmark & \cmark \\
Parcae~\cite{parcae} & Online (Migration) & \xmark & \cmark & \cmark \\
TrainMover~\cite{trainmover} & Online (Migration) & \xmark & \cmark & \cmark \\
\textbf{\sys} & \textbf{Online (Live)} & \textbf{\cmark\,(\textit{Any TP/PP/DP})} & \textbf{\cmark} & \textbf{\cmark} \\ \bottomrule
\multicolumn{5}{l}{ $^*$ Supports limited reshaping within fixed templates.}
\end{tabular}%
}
\end{table}

In summary, existing work typically optimizes either restart efficiency or online continuity, but not both together with \textit{arbitrary} mixed-parallel reshaping.
\sys combines all three properties required by highly volatile LLM training: storage-free reconfiguration, init-free live transition, and topology-flexible reshaping across arbitrary parallelism strategies.

%% file: motivation.tex
\section{Motivation}
\label{sec:motivation}
For cost-effective use of volatile resources, existing approaches force a trade-off between reconfiguration speed and parallelism flexibility. To break this trade-off, a live reconfiguration system should satisfy three competing demands simultaneously: it should avoid stopping the active training world, it should not require a full model replica in memory, and it should reshape state across arbitrary TP/PP/DP layouts. No existing system satisfies all three.

We now articulate the three fundamental challenges that make this combination difficult, which motivate the design of \sys.

\textbf{(1) Challenge: Constructing new topologies without stopping. } 
Distributed training frameworks conventionally bind the training loop to a single global process group. Reconfiguring the topology (adding/removing GPUs, changing TP/PP/DP) requires tearing down the old group and synchronously building a new one—a Stop-and-Restart event. To avoid paying this cost on the critical path, the system must decouple the preparation of the new topology from the execution of the old one. This requires maintaining two logical generations of the world simultaneously, with a clean handoff boundary between them. The handoff must be atomic and must not require materializing a second full copy of the model.

\textbf{(2) Challenge: Hiding distributed initialization latency. } 
When scaling out, newly added GPUs are \emph{cold}: they must initialize CUDA contexts, compile JIT kernels, and autotune communication primitives etc. In a naive design, the training stops, these cold nodes join the global process group immediately, and any collective operation blocks until all ranks reach the same synchronization point. The challenge is to let new workers complete heavyweight local setup without participating in global collectives, while avoiding stop the training process and to seamlessly integrate them into the training loop only after they are fully warmed up.

\textbf{(3) Challenge: Decoupling state from parallelism strategy. } 
When the parallelism configuration changes from $(TP, PP, DP)$ to $(TP', PP', DP')$, the mapping of tensor shards to ranks is invalidated. For example, a weight matrix that was split across 4 TP ranks must now be split across 8. The system must remap and move large model/optimizer states while preserving correctness and bounded memory. This is exactly where many online approaches restrict themselves to fixed templates (e.g., adjusting only DP replica count), because computing precise tensor-shard intersections across arbitrary TP/PP/DP decompositions is more complex than plain restart. The challenge is to perform this transfer efficiently, correctly, and without materializing a full model replica.

We motivate \sys through three empirical observations that expose the limitations of existing approaches and establish the need for a live reconfiguration runtime with arbitrary mixed-parallel reshaping, we now present the design of \sys, which addresses each through a carefully co-designed set of mechanisms: dual-world execution with shadow initialization (Section~\ref{sec:design-shadow}), mock process groups for isolated warmup (Section~\ref{sec:design-isolated}), and a layer-streaming resharding protocol with bounded memory (Section~\ref{sec:design-reshard}).

%% file: design.tex
\section{Design of \sys}
\label{sec:design}

We present the design of \sys, a live reconfiguration runtime that treats elastic scaling as a bounded network operation rather than a restart-driven recovery workflow.
The central insight is that \textbf{the two dominant costs of Stop-and-Restart---storage I/O and distributed re-initialization---can be reduced and be hidden if the reconfiguration path never tears down the running world}.
 And before detailing each mechanism, we first define the system model and walk through the end-to-end workflow.

\subsection{System Model and Scope}
\label{sec:design-event-model}

\sys targets \emph{warning-based or planned elasticity} events: scheduled maintenance windows, scheduler-driven scale changes, and preemptions with advance notice (e.g., spot instance eviction warnings, which typically provide a 2-minute grace period~\cite{aws-spot}).
Under this model, all participating GPUs remain reachable during the bounded transfer window (seconds, as we will show in Section~\ref{sec:evaluation}).

\paragraph{Elasticity event spectrum.}
Elasticity events in practice span a continuous spectrum rather than a binary division. At one end are \emph{planned resizes} (scheduler-driven, arbitrarily long warning window), where \sys has ample time to complete the full handoff. In the middle are \emph{preemption warnings} (e.g., AWS Spot's 2-minute notice~\cite{aws-spot}), where the warning window comfortably exceeds the preparation time (Section~\ref{sec:eval-breakdown}). At the far end are \emph{unannounced fail-stop} events (sudden power loss, kernel panic), where GPU state becomes unreachable and no live mechanism can help. Public trace data from major cloud providers indicates that the vast majority of Spot interruptions provide advance notice~\cite{aws-spot, skypilot}; unannounced failures are comparatively rare and are already well-served by existing checkpoint-recovery mechanisms.


\begin{figure*}[t]
    \centering
    \begin{minipage}[t]{0.48\linewidth}
        \centering
        \includegraphics[width=\linewidth]{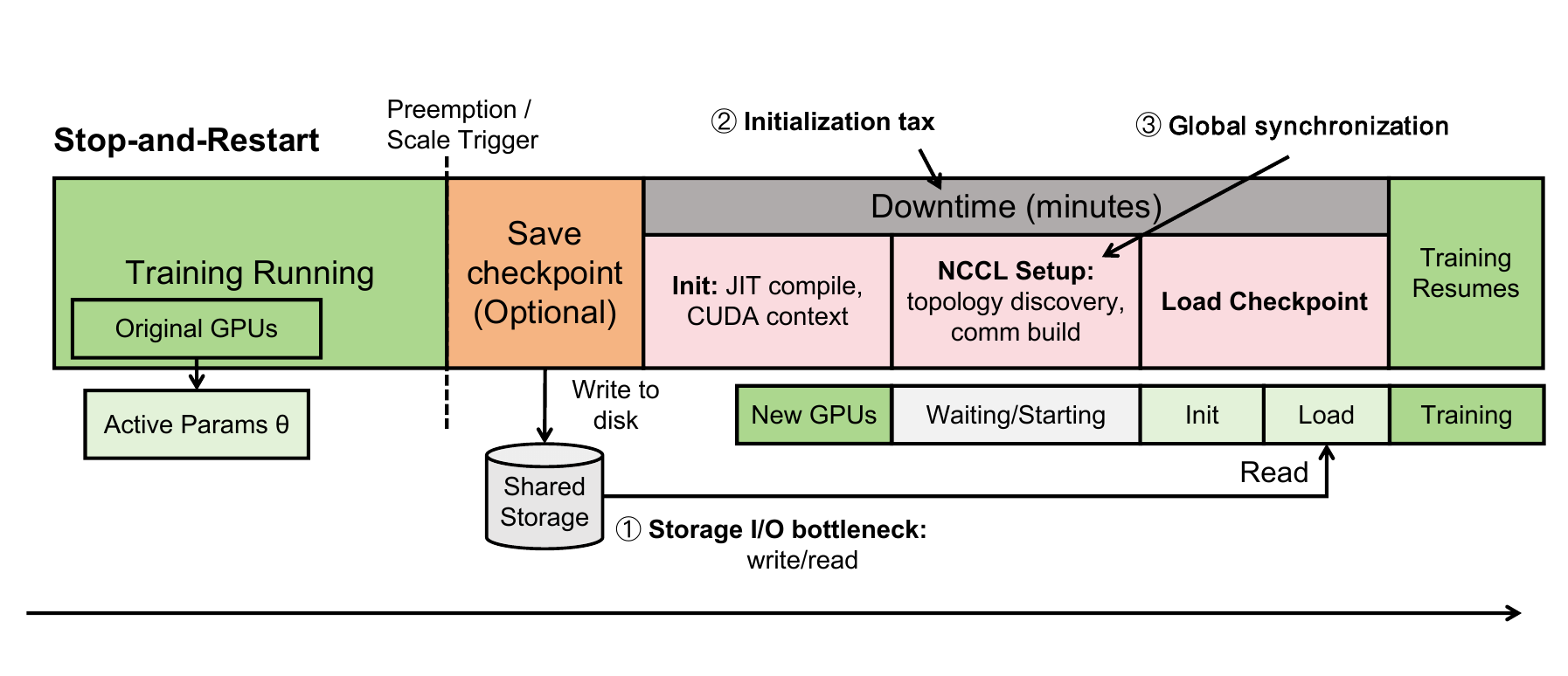} \\[3pt]
        \textbf{(a)} Stop-and-Restart
    \end{minipage}
    \hfill
    \begin{minipage}[t]{0.48\linewidth}
        \centering
        \includegraphics[width=\linewidth]{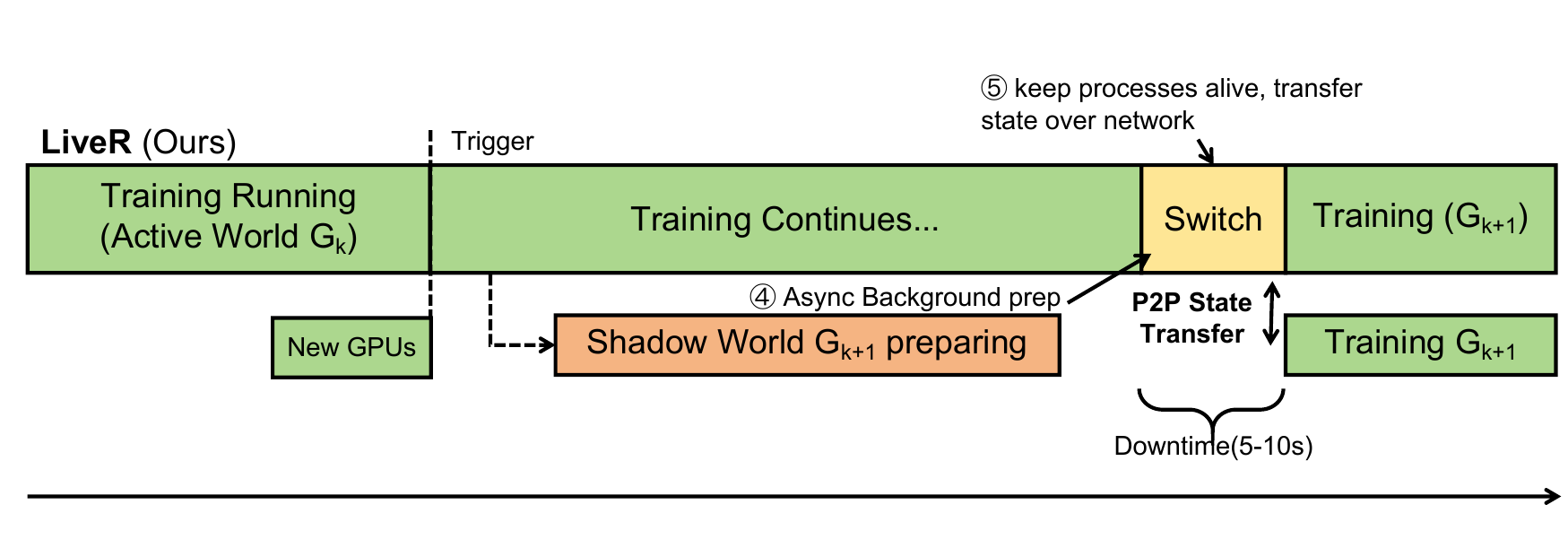} \\[3pt]
        \textbf{(b)} \sys
    \end{minipage}
    \caption{Reconfiguration sequence comparison. Stop-and-Restart (left) incurs minutes of downtime due to storage I/O and distributed re-initialization. \sys (right) reduced both bottlenecks by keeping processes alive and streaming state directly over high-bandwidth interconnects.}
    \label{fig:reconfig-compare}
\end{figure*}

\paragraph{Graceful degradation at the boundary.}
Between these extremes lies a narrow \emph{gray zone}: events with very short or no notice where a full handoff cannot complete, but GPU state is still momentarily reachable. \sys degrades gracefully in this regime. Because the Active World's parameter storage remains intact and consistent until the Atomic Switch commits (invariant I4), an interrupted handoff can fall back to the latest durable checkpoint with no loss beyond the normal checkpoint interval. Moreover, any partial progress from the Shadow World (e.g., pre-built communicators, completed mock warmup) can accelerate the subsequent checkpoint-based recovery, since the new topology is already partially constructed. This design makes \sys's live path a \emph{strict improvement} over checkpoint-only recovery: when the warning window suffices, it eliminates I/O and initialization entirely; when it does not, the system degrades to the same checkpoint baseline, potentially faster due to residual Shadow World work.

As we discuss in Section~\ref{sec:bg-volatility}, warning-based events constitute the majority of elasticity triggers in production spot-instance settings; unexpected fail-stop losses are comparatively rare and are already well-served by existing checkpoint-recovery mechanisms.

\subsection{Design Invariants}
\label{sec:design-invariants}

All mechanisms in \sys are governed by four invariants that together define the correctness and resource boundaries of the system:

\textbf{I1---No global restart on the live path.} For warning-based or planned elasticity events, existing ranks continue training while the new groups are prepared in the background.
\textbf{I2---Bounded memory during transition.} \sys does not keep a second full copy of model states; the additional memory overhead is limited to metadata plus a fixed-size staging buffer.
\textbf{I3---Deterministic switch boundary.} Reconfiguration is committed only at a consistent iteration boundary after in-flight work has drained.
\textbf{I4---Explicit failure boundary.} Unexpected fail-stop loss falls outside the live path and triggers checkpoint-based recovery.


\subsection{End-to-End Workflow}
\label{sec:design-e2e}

\begin{figure*}[t]
    \centering
    \includegraphics[width=\textwidth]{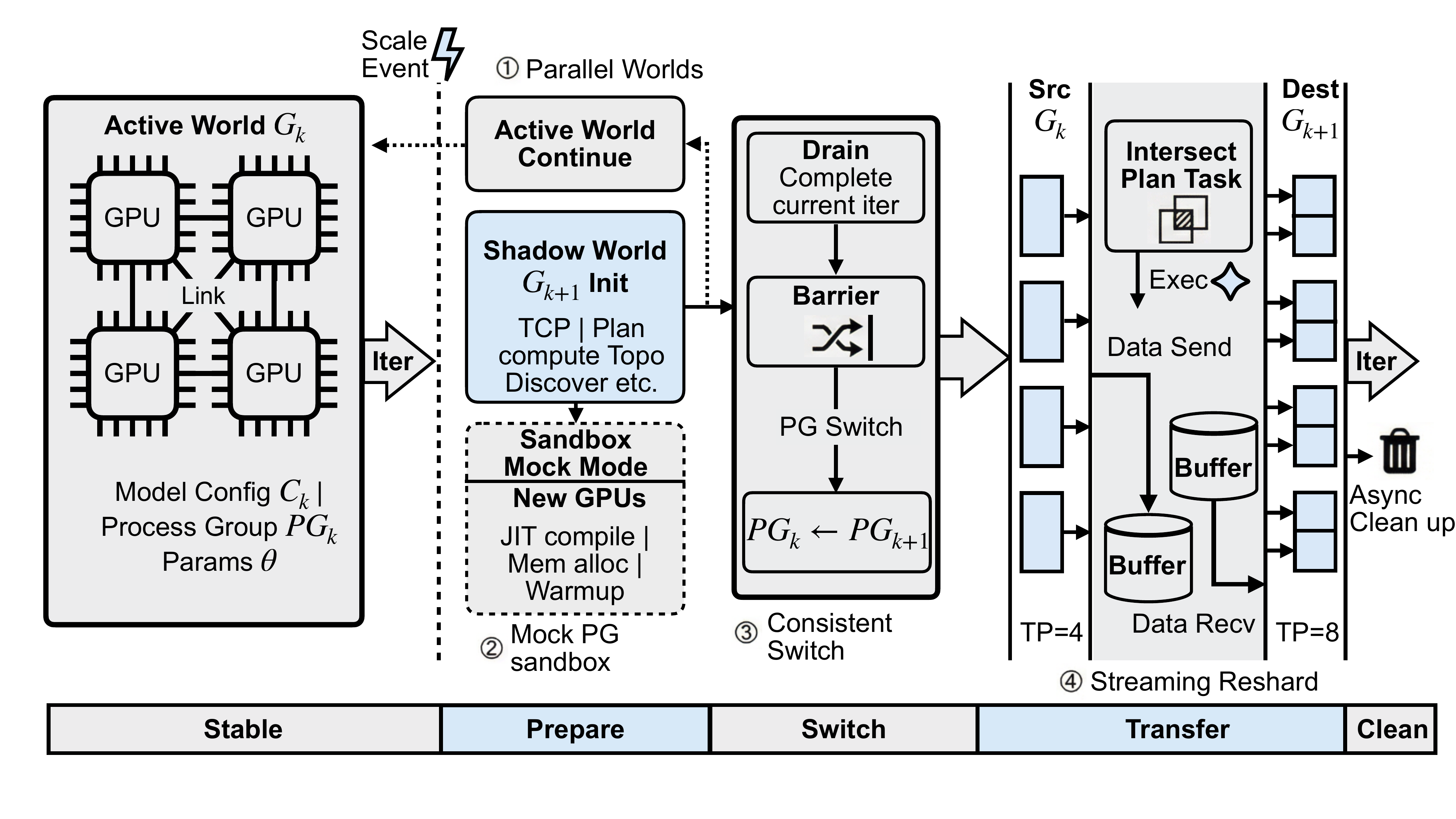}
    \caption{\sys architecture overview. The \emph{foreground plane} executes the training loop using the Active World's process groups; the \emph{background plane} handles Shadow World initialization and Streaming Resharding state transfer. The Atomic Switch is a sub-second metadata swap; no second model copy is ever instantiated.}
    \label{fig:arch}
\end{figure*}

Figure~\ref{fig:reconfig-compare} shows \sys timeline compared to stop-and-restart method.
Figure~\ref{fig:arch} illustrates the \sys architecture. 
The system employs a dual-plane design: the \emph{foreground plane} executes the training loop using the Active World's process groups, while the \emph{background plane} handles shadow initialization and state migration.
Reconfiguration proceeds without blocking the critical training path. 

The full lifecycle of a reconfiguration event proceeds in eight steps. First, the controller receives a \emph{trigger}---a resize event such as a spot eviction warning or a scheduler scale-out command---and determines the target topology $(TP', PP', DP')$. While the Active World ($G_k$) continues training, the background plane constructs the \emph{Shadow World} ($G_{k+1}$): new processes are spawned, TCP bootstrap and NCCL topology discovery run on separate threads, and communicators for the target configuration are established (Section~\ref{sec:design-shadow}). Second, new ranks complete heavyweight local setup---model construction, CUDA context, JIT compilation, kernel autotuning---inside a sandboxed environment that intercepts all collective calls, preventing cold nodes from blocking hot ones (Section~\ref{sec:design-isolated}). The controller computes the geometric intersection of tensor shards between the old and new configurations, generating a minimal peer-to-peer \emph{transfer plan} (Section~\ref{sec:design-reshard}). Third, the Active World finishes iteration $N$ and drains in-flight workloads, yielding a clean boundary with no residual dependencies, process-group references and routing metadata are atomically redirected from $G_k$ to $G_{k+1}$, this is a sub-second metadata swap. Fourth, parameter shards move layer-by-layer through a fixed staging buffer over high-bandwidth interconnects, directly from source ranks to destination ranks without touching storage (Section~\ref{sec:design-reshard}), and no second copy of the model is instantiated (Section~\ref{sec:design-switch}). Finally, old-world communicators and staging buffers are reclaimed asynchronously while iteration $N{+}1$ starts immediately on the new world.

Steps 1--2 overlap entirely with ongoing training (invariant I1).
Step 3 constitute the only window where training pauses; as we show in Section~\ref{sec:evaluation}, this window lasts seconds, not minutes.

\subsection{Parallel Worlds: Constructing New Topologies Without Stopping}
\label{sec:design-shadow}

\textbf{Challenge.}
Distributed training frameworks bind the training loop to a single global set of process groups.
Reconfiguring the topology---adding or removing nodes, changing the TP/PP/DP decomposition---conventionally requires tearing down the old groups and synchronously building new ones.
This teardown-rebuild cycle forces a full Stop-and-Restart, violating invariant I1.

\textbf{Approach.}
\sys decouples execution from preparation by maintaining two logical generations of the world simultaneously. The \emph{Active World} ($G_k$) is the currently executing configuration: all training kernels and NCCL collectives use $G_k$'s communicators. The \emph{Shadow World} ($G_{k+1}$) is the target configuration, constructed entirely in the background.

While $G_k$ continues to compute, $G_{k+1}$ performs TCP bootstrap, NCCL topology discovery, and communicator setup on separate threads and new devices.
These operations are inherently asynchronous---they involve network handshakes and control-plane coordination but no GPU computation---so they can run concurrently with the training loop without contending for GPU computation resources.

The key property is that all expensive control-plane setup is shifted off the critical training path.
When the Shadow World is fully initialized, the system transitions to the \emph{Ready} state and waits for the next iteration boundary to proceed with transfer and switch.

\subsection{Mock Process Groups: Hiding Initialization Latency}
\label{sec:design-isolated}

\textbf{Challenge.}
When scaling out, new nodes are \emph{cold}: they must boot processes, initialize CUDA contexts, compile JIT kernels, and autotune communication primitives.
In a naive design, these new ranks join the global process group immediately, and any collective operation they participate in blocks until \emph{all} ranks---old and new---reach the same synchronization point.
A single joiner can stall the entire process,
this is precisely the distributed re-initialization overhead identified in Section~\ref{sec:bg-overhead}.

\textbf{Approach.}
\sys interposes a \emph{Mock Process Group} layer between the new ranks and the real NCCL communicators.
In mock mode, all collective calls (e.g., \texttt{AllReduce}, \texttt{Broadcast}, \texttt{Barrier}) are intercepted and return locally without performing actual network communication.
This allows new ranks to complete model construction and parameter allocation, trigger JIT compilation of forward/backward kernels, and run kernel autotuning and memory planning---all of which are \emph{local} operations that require no global coordination.
Once the mock warmup is complete, the rank discards its mock state and joins the real Shadow World communicators.

The critical property is a \emph{symmetry break}: local initialization and global coordination are decoupled, so active ranks never wait for cold-start latencies.
This directly addresses the distributed initialization overhead that checkpointing-based approaches cannot eliminate.

\subsubsection{Companion Manager and Generation State Machine}
\label{sec:design-shadow-companion}

To coordinate the dual-world lifecycle without race conditions, each rank is paired with a lightweight \emph{Companion Manager}---a daemon thread that handles control-plane tasks: peer discovery, background handshakes, communicator lifecycle management, and state transitions. They run concurrently with GPU-bound training kernels and introduce negligible interference (quantified in Section~\ref{sec:eval-breakdown}).

Each configuration is tagged with a monotonic \emph{generation id}.
The state machine transitions through five states. In the \emph{Stable} state, only the Active World exists and training proceeds normally. Upon a resize trigger, the machine enters the \emph{Prepare} state, which initiates Shadow World bootstrap and mock warmup. When the Shadow World is fully initialized, the system transitions to the \emph{Ready} state and awaits the next iteration boundary. At the consistent cut, the machine enters the \emph{Switch} state, original training progress stalls as process group references are atomically redirected and parameters are transferred to the new GPUs. Finally, the \emph{Cleanup} state reclaims old-generation resources asynchronously.

Figure~\ref{fig:state_machine} summarizes these transitions.
The generation id ensures that stale references to old communicators are never used after a switch, and the monotonic ordering guarantees that at most two generations coexist at any time (invariant I2).

\begin{figure}[h]
    \centering
    \includegraphics[width=\linewidth]{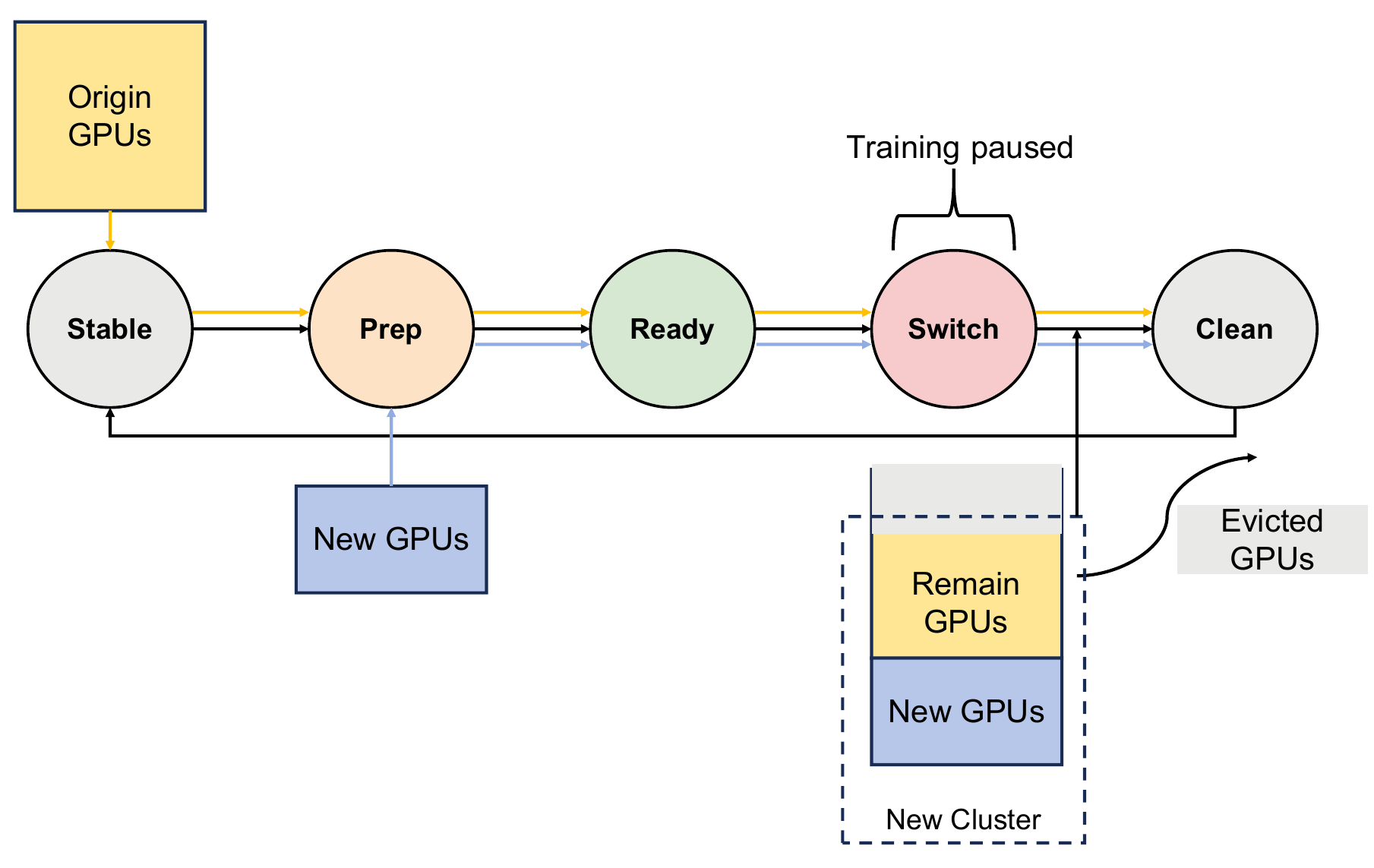}
    \caption{Generation state machine for safe world transitions. Each configuration is identified by a monotonic generation id; at most two generations coexist.}
    \label{fig:state_machine}
\end{figure}

\subsection{Streaming Resharding: Arbitrary State Transfer with Bounded Memory}
\label{sec:design-reshard}

\textbf{Challenge.}
When the parallelism configuration changes from $(TP, PP, DP)$ to $(TP', PP', DP')$, the mapping of tensor shards to ranks is invalidated. For example, a weight matrix that was split across 4 TP ranks must now be split across 8.
This state transfer is the core difficulty of \emph{arbitrary} reshaping---it is exactly why existing online systems restrict themselves to fixed templates (e.g., adjusting only DP replica counts).

Two constraints make this problem hard:
(1) the transfer must be \emph{correct}---every element of the global state must appear exactly once in the new configuration; and
(2) the transfer must be \emph{memory-bounded}---no single GPU can materialize a full model replica during the transition (invariant I2).

\textbf{Approach.}
\sys solves this with a two-phase design: \emph{metadata-driven transfer planning} followed by \emph{bounded layer-streaming execution}.

\subsubsection{Abstract Resource View and Transfer Planning}

Standard frameworks hard-code the mapping from tensor shards to rank indices (e.g., rank $k$ owns the $k$-th column slice).
\sys replaces this with an \emph{Abstract Resource View}: the training state is modeled as logical tensors (e.g., weights of layer $L$) plus a sharding specification, independent of physical rank assignment.

Given the old configuration $C_{k}$ and the new configuration $C_{k+1}$, the planner computes, for each tensor, the \emph{geometric intersection} of source and destination shard views.
This intersection directly yields a set of \texttt{TransferTask} objects specifying the exact byte ranges to be sent from each source rank to each destination rank.

For example, in a transition from $TP=4$ to $TP=8$, each source rank owns $W[:, 0{:}N/4]$ and each destination rank needs $W[:, 0{:}N/8]$.
The intersection computation identifies that each source rank sends the first half of its columns to one destination rank and the second half to another---no redundant data movement, no full-tensor materialization.
Figure~\ref{fig:resharding} visualizes this geometry.

\begin{figure}[t]
    \centering
    \includegraphics[ width=1.0\linewidth]{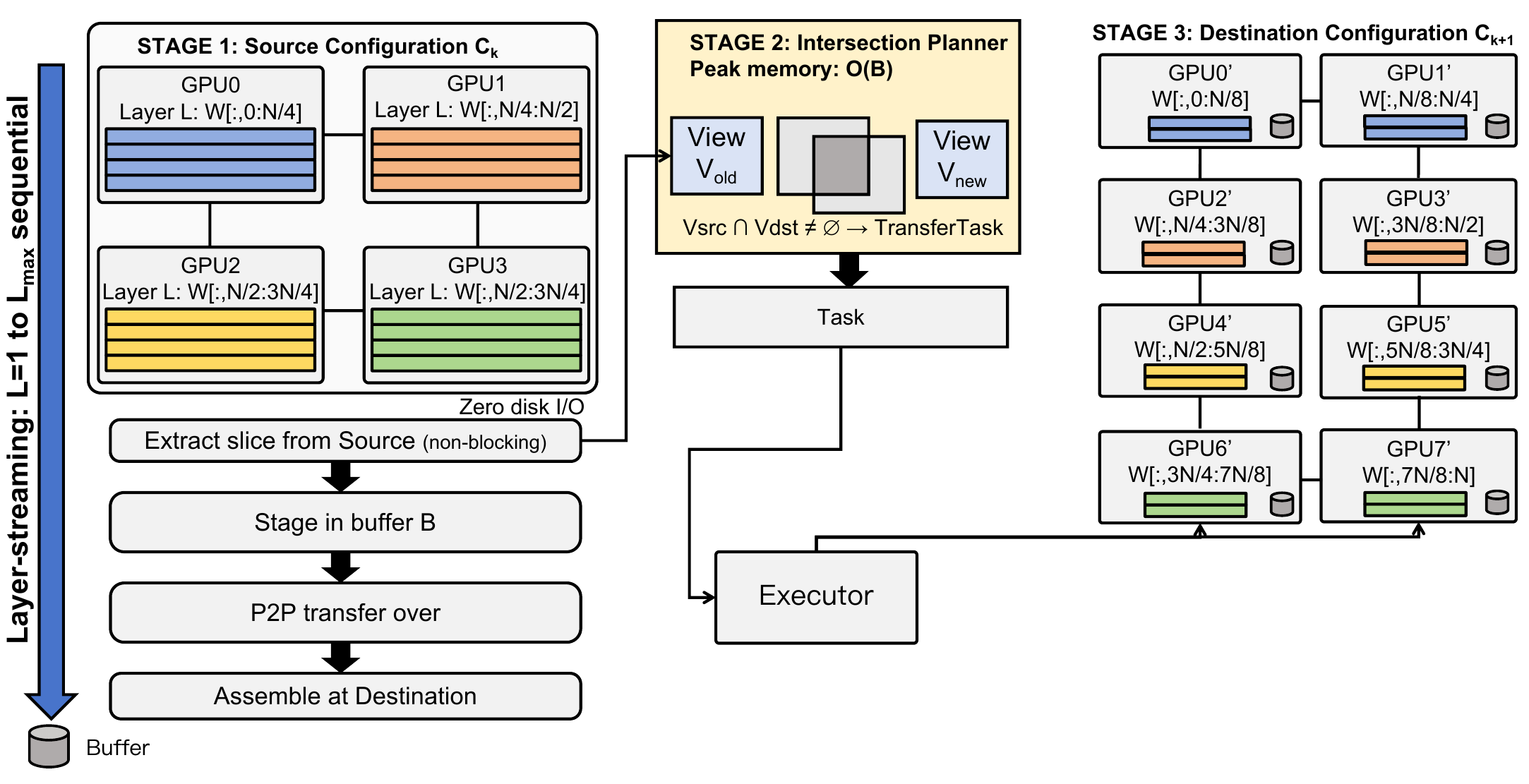}
    \caption{Intersection-based transfer planning for TP resharding ($TP{=}4 \to TP{=}8$). Each source rank sends half its shard to two destination ranks. The intersection defines the exact byte ranges---no full-tensor materialization is needed.}
    \label{fig:resharding}
\end{figure}

This approach generalizes across all parallelism dimensions. For TP transitions, column or row slices are split or merged along the sharded dimension. For DP transitions, increasing replicas degenerates to a broadcast pattern while decreasing replicas discards redundant copies. For PP transitions, entire layers move between stages; the intersection is either the full tensor or empty.

The planning phase runs entirely on CPU and completes in under 1 second even for a 175B-parameter model with 96 layers and 1024 ranks, as the computation involves only sharding metadata, not actual tensor data.

\subsubsection{Layer-Streaming Protocol}

Executing all transfer tasks simultaneously would materialize the entire model state in the staging buffers, violating invariant I2.
\sys therefore executes transfers \emph{one layer at a time} through a fixed-size staging buffer $B$ (typically 512\,MB--1\,GB). For each layer $\ell$, source ranks send the relevant slices to destination ranks via peer-to-peer NCCL connections; destination ranks receive incoming slices into the staging buffer, assemble the new shard, and write it to the model's parameter storage. The staging buffer is reused across layers, capping the additional memory at $O(B)$ regardless of model size.

This protocol ensures that the transition overhead never scales with total model size---only with the per-layer state size and the staging buffer budget.
For full formalization of the configuration space, correctness conditions, and protocol-level pseudo-code, see Appendix~\ref{sec:appendix-transfer}.

\subsection{Switch and Consistent Cut}
\label{sec:design-switch}

The exact sequence at the reconfiguration boundary proceeds in three phases. For a scale out event, first, active ranks \emph{drain}: they finish iteration $N$, and under 1F1B and interleaved pipeline schedules~\cite{pipedream, megatron-lm}, the end of an iteration is a natural consistent cut where GPU in the current step have completed forward and backward passes with no in-flight dependencies. Second, the Streaming Resharding protocol executes bounded layer-by-layer \emph{transfer} of persistent state from the old configuration to the new one; training is paused during this window. Third, process-group references and routing metadata are atomically \emph{switched}: all ranks redirect their communication endpoints from $G_k$ to $G_{k+1}$, a sub-second operation involving only pointer swaps and metadata updates with no data movement.

After the switch, iteration $N{+}1$ starts directly on the new world with pre-warmed communicators and pre-compiled kernels.
No process restart, no NCCL rebuild, no CUDA re-initialization is required.

\subsection{Resource and Overhead Boundaries}
\label{sec:design-overhead}

\paragraph{Memory overhead.}
The additional memory during reconfiguration consists of NCCL communicator metadata for the Shadow World (a few MB per group) and the staging buffer $B$ (typically 512\,MB--1\,GB, configurable based on available memory).
The total overhead is bounded by $O(B + C)$, where $C$ is the communicator metadata.
For a 14B model on 80\,GB A800 GPUs, this amounts to less than 2\,GB---well within the headroom of typical training configurations.
Detailed accounting appears in Appendix~\ref{sec:appendix-mem}.

\paragraph{Failure fallback.}
If a fail-stop loss occurs before the transfer completes, the system has not yet committed the switch---the old configuration's parameter storage remains intact and consistent at iteration $N$.
\sys falls back to the latest durable checkpoint (which may be iteration $N$ or earlier, depending on the checkpoint interval) and resumes via Stop-and-Restart.
This satisfies invariant I4 and ensures correctness without requiring in-memory replication.

%% file: implementation.tex
\section{Implementation}
\label{sec:implementation}

We implement \sys strictly as a runtime extension to Megatron-LM and PyTorch Distributed. The implementation comprises approximately 10,000 lines of Python code and introduces minimal intrusiveness: the core training loop in Megatron remains largely unmodified, save for hooks that invoke the central controller and runtime manager at iteration boundaries.

\paragraph{Shadow Group and Companion Manager.} 
The \sys Companion Manager is implemented using Python. It runs a daemon thread alongside the main GPU worker. For the Active World, it delegates to the standard implementation; for the Shadow World, it diverts initialization to a separate CUDA stream and CPU thread, thus ensuring that TCP handshakes, socket establishment, and subsequent transport connection functions do not stall the CPU thread dispatching kernels to the active GPU stream.

\paragraph{Streaming Resharding.} 
The intersection-based resharding planner traverses the computational graph dynamically using PyTorch's tensor metadata. To execute the data transfers without exceeding memory, \sys partitions model state into fixed-size chunks (default $B=512\,\text{MB}$). It leverages PyTorch's \texttt{TCPStore} to exchange port configurations, then establishes direct peer-to-peer connections between old and new ranks. The transfer is pipelined using \texttt{isend()} and \texttt{irecv()}.  

\paragraph{Simulator Calibration.} 
To evaluate \sys at the scale of thousands of GPUs (e.g., 70B parameter models on 1,024 GPUs) and conduct steady-state analyses, we built a discrete-event simulator using the SimPy framework. Crucially, the simulator is rigorously calibrated against physical profiling traces from our testbed. Network transfer times are modeled based on profiled NCCL \texttt{all\_to\_all} and point-to-point bandwidth under contention; computation steps are derived from Megatron-LM iteration latencies; and distributed initialization overheads (TCP/NCCL setup) are measured empirically across different cluster sizes using PyTorch By strictly separating measured results from extrapolated metrics, we ensure that simulator-based claims for large-scale deployments are solidly anchored in real-world performance data.

%% file: evaluation.tex
\section{Evaluation}
\label{sec:evaluation}
We evaluate \sys on a physical testbed along four dimensions: end-to-end reconfiguration speedup over checkpoint-based baselines (\S\ref{sec:eval-speedup}), detailed breakdown of the live reconfiguration pipeline and its steady-state overhead (\S\ref{sec:eval-breakdown}),
and end-to-end training efficiency under frequent resource volatility (\S\ref{sec:eval-goodput}). Numerical correctness is verified in \S\ref{sec:eval-correctness}. Large-scale simulation results (up to 1,024 GPUs, 70B parameters) appear in \S\ref{sec:eval-simulation}.

\subsection{Experimental Setup}
\label{sec:eval-setup}

\paragraph{Testbed.}
All experiments run on a cluster of four NVIDIA A800 (80\,GB) PCIe servers (NF5468M6), providing 32 GPUs in total. Each node is equipped with dual-socket Intel Xeon 8358 CPUs (2.6\,GHz, 32 cores each), 1\,TB of DDR4 DRAM, and two 3.84\,TB NVMe SSDs for local storage. Intra-node GPU communication uses the PCIe Gen4 bus, while inter-node traffic traverses a non-blocking 200\,Gbps HDR InfiniBand fabric (one dual-port ConnectX-6 HCA per node). Each node also has a dual-port 25\,Gbps Ethernet NIC for management traffic. \sys is implemented on top of Megatron-LM and PyTorch, with NCCL handling collective communication. 

\paragraph{Baselines.}
We compare \sys against two baselines that cover the design space of elastic training, categorized by whether they support \emph{arbitrary parallelism reshaping}, avoid \emph{process restart}, and reduce \emph{persistent storage I/O}.

\begin{itemize}
    \item \emph{Megatron-LM Checkpoint}~\cite{megatron-lm} represents standard production practice: training pauses, persists the full distributed state to shared storage(optional, here we choose to fallback to previous checkpoint, so there is no saving time), reloads, and performs a full process restart on the new topology.

    \item \emph{UCP}~\cite{universal} and \emph{ByteCheckpoint}~\cite{bcp} represent the state of the art in optimized checkpointing. Both support \emph{parallelism reshaping during reload} (e.g., load-time resharding), which narrows the reload latency gap, yet both still require a full process restart and depend on persistent storage to reconstruct the distributed state. We model them by reducing reload latency while retaining the restart and re-initialization overhead, We use UCP as the representative checkpoint-reshape baseline.
\end{itemize}



\begin{figure}[t]
    \centering
    \includegraphics[width=\linewidth]{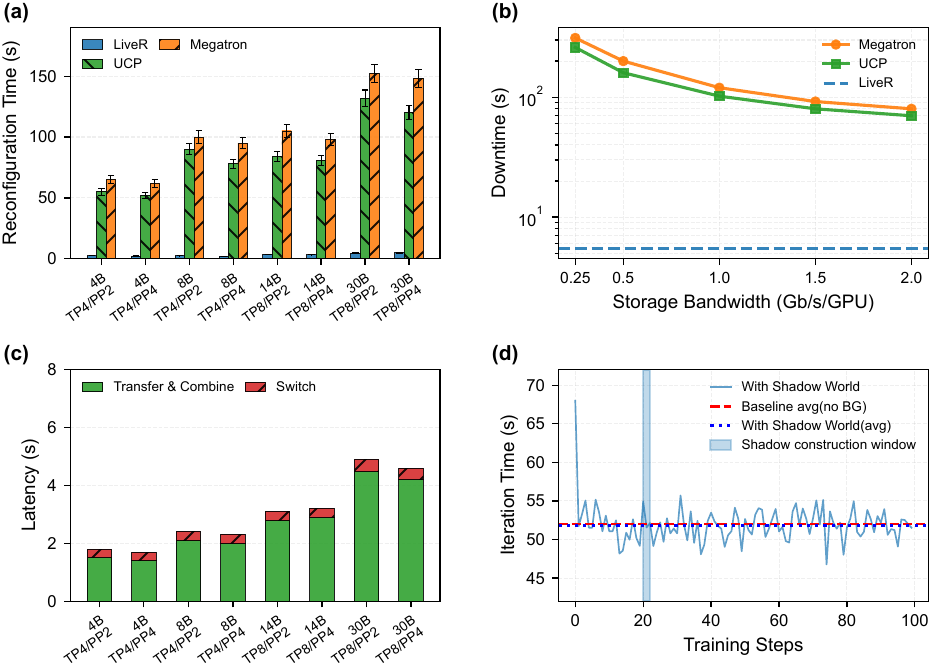}
    \caption{End-to-end evaluation on the 32-GPU A800 testbed. (a)~Reconfiguration downtime across model sizes: \sys achieves 14$\times$--23$\times$ speedup over Megatron-LM Checkpoint and consistently stays below 6\,s. (b)~Storage bandwidth sensitivity for GPT-14B: checkpoint-based systems degrade sharply at low bandwidth, while \sys is storage-independent. (c)~Latency breakdown of a live reconfiguration event:  Transfer And Combine grows with model size and remains a majority at 32-GPU scale; Switch is a sub-second metadata swap. (d)~Steady-state iteration time during concurrent Shadow World initialization: no latency spikes or throughput drops are observed across 100 steps.}
    \label{fig:eval_main}
\end{figure}

\subsection{Reconfiguration Speedup}
\label{sec:eval-speedup}

Figure~\ref{fig:eval_main}(a) compares the end-to-end reconfiguration downtime of \sys against both checkpoint baselines across all four model sizes and multiple parallelism settings.

For Megatron-LM Checkpoint, latency is dominated by two factors: reloading distributed states, and the cold start penalty of re-spawning processes, rebuilding NCCL topologies, and JIT compilation etc. For GPT-30B model, this translates to nearly 150\,s of forced idle time per scaling event. UCP reduces the checkpoint load latency but process launch, framework initialization and warmup still remain on the critical path, keeping downtime above 120\,s.

\sys bypasses both bottlenecks entirely. By streaming state over PCIe/InfiniBand and hiding initialization in the Shadow World, it achieves \textbf{14$\times$--23$\times$ speedup} over Megatron-LM Checkpoint, reducing downtime to the 2--6\,s range across all model sizes, background initialization shifts the startup cost off the critical path.


\paragraph{Impact of Storage Bandwidth.}
Checkpoint-based systems are fundamentally sensitive to storage bandwidth, while \sys is not. We evaluate Megatron-LM Checkpoint under checkpoint bandwidths ranging from 0.25 to 2.0\,Gb/s per GPU. As shown in Figure~\ref{fig:eval_main}(b), at 0.25\,Gb/s per GPU, checkpoint loading alone takes over 300\,s for GPT-14B. In contrast, \sys streams state over the InfiniBand interconnect, skipping storage I/O. Its reconfiguration time is therefore insensitive to storage bandwidth variations.

\subsection{Breakdown and Overhead}
\label{sec:eval-breakdown}

To quantify the components of \sys's downtime, we decompose the reconfiguration latency into two phases that lie on the critical path, as shown in Figure~\ref{fig:eval_main}(c). Shadow communicator bootstrap, framework initialization and intersection-based transfer planning etc. are intentionally excluded: they run in background concurrent with training and are fully overlapped before the reconfiguration event stall period.

The Switch phase ($<$0.5\,s) atomically redirects process-group references and routing metadata, while old-world resources are reclaimed asynchronously, then the Transfer And Combine phase ($\sim$2--4\,s) performs layer-by-layer state relocation through the staging buffer; with 14B parameters on A800 PCIe/IB, transferring the $\sim$28\,GB state requires around 2\,s, and this is the only phase that scales with model size.

\paragraph{Steady-State Interference.}
A natural concern is whether Shadow World initialization impacts the ongoing training loop. We profile 100 consecutive iterations with concurrent Shadow Group construction. The mean iteration step time varied by only 0.28\%  compared to the baseline without any background activity. The Companion Manager's thread pool will handle the creation of new process groups and data transfer. As shown in Figure~\ref{fig:eval_main}(d), no significant latency spikes or throughput drops are observed during background construction, confirming that Shadow communicators built on separate CUDA streams and CPU threads, along with mock-mode warmup that intercepts all collective calls locally, introduce negligible interference.

\subsection{Sensitivity to Volatility Regime}
\label{sec:eval-volatility}

The goodput advantage of \sys depends on how frequently scaling events occur. To quantify this, we repeat the 8-hour volatility experiment from \S\ref{sec:eval-goodput} under three regimes: low (60-min intervals, $\sim$8 events), medium (30-min intervals, $\sim$16 events), and high (10-min intervals, $\sim$48 events). Figure~\ref{fig:volatility} reports the aggregate training efficiency(GPU Utilization) for each regime across all three systems.

Under the low-volatility regime, all systems achieve 95--99\% efficiency because reconfiguration overhead is amortized over long stable periods; even Megatron-LM Checkpoint reaches 95.2\%. As volatility increases, the gap widens sharply: at middle volatility, Megatron-LM Checkpoint drops to 79.8\%, while UCP reaches 85.6\%; at high volatility, Megatron-LM Checkpoint drops to 58.2\%, while UCP reaches 61.3\%. \sys sustains 99.1\% even under the most aggressive regime, because each live reconfiguration costs seconds rather than minutes and simultaneously optimizes the parallelism configuration.

\begin{figure}[t]
    \centering
    \includegraphics[width=\linewidth]{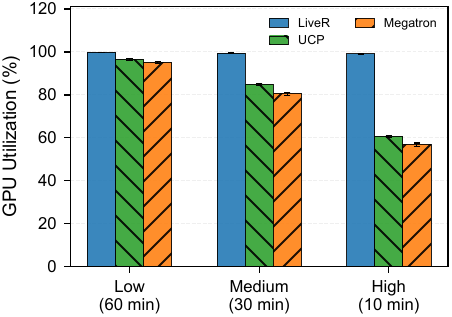}
    \caption{Training efficiency across three volatility regimes (GPT-14B, 32 GPUs, 8 hours). \sys's advantage grows with event frequency: from negligible at low volatility to a 19.3 percentage-point lead over Megatron-LM Checkpoint at high volatility.}
    \label{fig:volatility}
\end{figure}

\subsection{Training Efficiency under Volatility}
\label{sec:eval-goodput}

We benchmark \sys's end-to-end resilience by subjecting the GPT-14B model to a high-volatility environment on our 32-GPU testbed. We inject scale-out and scale-in events at intervals of 30--60 minutes, simulating spot instance availability patterns. Each event changes the active GPU count.

As plotted in Figure~\ref{fig:single_job}, summing the idle intervals over 47 distinct elasticity events accumulates to a loss of 80+ GPU-hours for the baseline. \sys slashes this wasted allocation down to a mere 4.1 GPU-hours. The Reconfiguration Downtime---the cumulative time spent paused during all 47 reconfigurations---is reduced from over 130 minutes (Megatron) and 100+ minutes (UCP) to just 7 minutes with \sys, a \textbf{14.2$\times$} improvement over the best baseline. 

This dramatic reduction stems from \sys's ability to maintain live training processes during elasticity events, avoiding the costly stop-and-restart cycle. While Megatron and UCP must tear down the entire training context, reinitialize NCCL, reload checkpoints from storage, and recompile CUDA kernels at every reconfiguration, \sys performs lightweight peer-to-peer state transfers in the background. Consequently, \sys achieves near-optimal goodput of \textbf{99.5\%}, compared to 93\% for UCP and 91\% for Megatron, effectively masking nearly all resilience overhead within the 24-hour training window.

\begin{figure}[t]
    \centering
    \includegraphics[width=\linewidth]{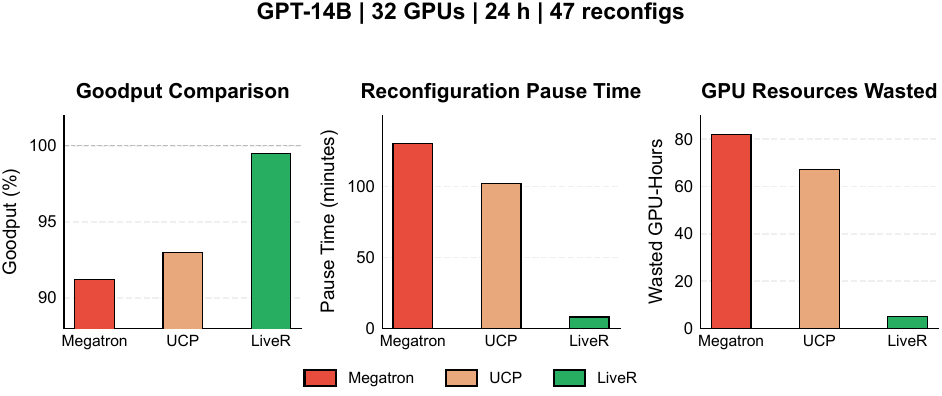}
    \caption{Cumulative wasted GPU-hours for a single training job reflecting directly on deployment costs.}
    \label{fig:single_job}
\end{figure}

\subsection{Correctness Verification}
\label{sec:eval-correctness}

Elasticity frameworks that manipulate distributed tensors risk numerical parity misalignment. We verify the bit-exactness of \sys's intersection-based transfer by triggering an arbitrary 3D reshape on the GPT-1.7B model at step 160: $(TP{=}2, PP{=}2) \to (TP{=}4, PP{=}1)$. After reconfiguration, we aggregate the global tensors and perform an element-wise differential trace against a statically mapped reference. The maximum deviation is exactly $\pm 0.0$.

\begin{figure}[t]
    \centering
    \includegraphics[ width=1.0\linewidth]{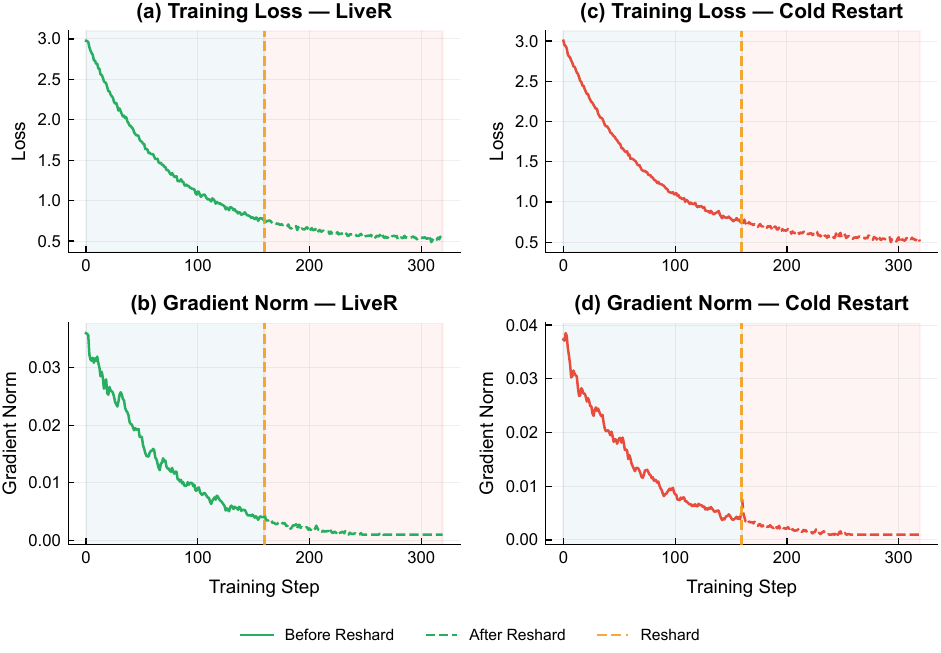}
    \caption{Training loss and gradient norm trace the static baseline identically before and after a live reshaping at step 160, confirming bit-exact numerical parity.}
    \label{fig:correctness}
\end{figure}

Figure~\ref{fig:correctness} shows that the loss trajectory and gradient norm are indistinguishable before and after the scaling event. The correctness guarantee follows from the completeness property of the intersection-based planning (Appendix~\ref{sec:appendix-transfer}): every element of the old configuration's tensor shards has a unique destination in the new configuration, and the union of all new shards is equivalent to the union of all old shards.

\subsection{Large-Scale Simulation Evaluation}
\label{sec:eval-simulation}
Our physical testbed is limited to 32 GPUs. To evaluate \sys at scales exceeding our hardware constraints (up to 1,024 GPUs and 70B parameters), we built a discrete-event simulator using the SimPy framework. This section describes the simulator design, calibration, and results.

\subsubsection{Simulator Design and Calibration}

The simulator faithfully models the critical paths of distributed LLM training: computation kernels (derived from roofline theoretical flops vs.\ empirical SM utilization), network topology, and persistent storage I/O latency distributions.

Crucially, the simulator is calibrated against our 32-GPU physical testbed. We directly profiled NCCL collective operations, TCP handshakes, and process startup times, feeding these constants into the simulator. Network transfer times are modeled based on profiled NCCL \texttt{all\_to\_all} and point-to-point bandwidth; computation steps are derived from Megatron-LM iteration latencies; and distributed initialization overheads (TCP/NCCL setup) are measured empirically across different cluster sizes using PyTorch 2.7.

\begin{figure}[t]
    \centering
    \includegraphics[width=\linewidth]{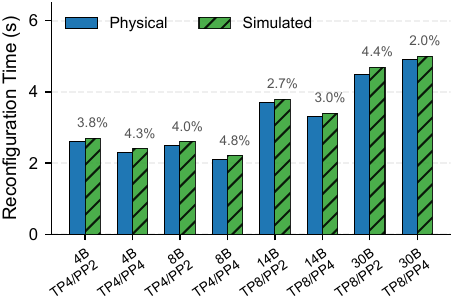}
    \caption{Simulator validation: physical vs.\ simulated reconfiguration latency across model sizes. The simulator achieves $<$5\% divergence from empirical measurements.}
    \label{fig:simulator_validation}
\end{figure}

Figure~\ref{fig:simulator_validation} validates the simulator by comparing physical measurements against simulated predictions for the same configurations on 32 GPUs. The divergence is less than 5\% across all model sizes, providing confidence in the large-scale predictions below.

\subsubsection{70B Model Reconfiguration}

\begin{figure}[t]
    \centering
    \includegraphics[width=\linewidth]{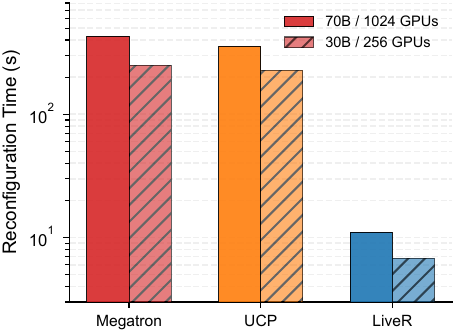}
    \caption{Simulated reconfiguration time for a 70B parameter model on 1,024 GPUs. \sys reduces downtime from $\sim$565\,s to $\sim$11\,s, a 50$\times$ improvement.}
    \label{fig:large_reshard}
\end{figure}

Figure~\ref{fig:large_reshard} extrapolates the reconfiguration comparison to a 70B parameter model on 1,024 GPUs. Cold Restart latency balloons linearly with network checkpoint volume and cluster initialization synchronization points, taking $\sim$565\,s. \sys reduces this to $\sim$11\,s---a \textbf{50$\times$ improvement}.

\subsection{Summary}

Across all physical experiments on the 32-GPU testbed, \sys achieves 14$\times$--23$\times$ speedup over Megatron-LM Checkpoint across model sizes and parallelism settings, with approximately 7\,s downtime and under 0.3\% steady-state overhead---confirming that background construction is non-intrusive.  And \sys sustains approximately 99\% training efficiency under frequent volatility versus 61.3\% for checkpoint-based baselines. Numerical parity is bit-exact across all tested reshaping scenarios. Large-scale results extrapolated to 1,024 GPUs via calibrated simulation further corroborate these findings.

%% file: conclusion.tex
\section{Discussion}
\label{sec:discussion}

While \sys demonstrates significant improvements over the Stop-and-Restart paradigm, several aspects regarding its operational boundaries and overheads warrant discussion.

\paragraph{Failure model and checkpoint fallback.}
\sys is designed explicitly for warning-based or planned elasticity (e.g., spot instance preemption notices, dynamic cluster resizing). It operates under the assumption that all GPUs remain reachable during the bounded transfer window (2--4\,s for models up to 30B). If an unexpected fail-stop loss occurs before the transfer is complete, \sys does not attempt to resolve the failure within the live path; instead, it falls back to the most recent durable checkpoint. As discussed in Section~\ref{sec:design-event-model}, this boundary is not a binary cliff but a graceful degradation: any partial Shadow World progress can accelerate the subsequent checkpoint-based recovery.

\paragraph{Network sensitivity and contention.}
\sys uses high-bandwidth interconnects (PCIe, InfiniBand) for state transfer. Consequently, its reconfiguration latency scales inversely with available bandwidth. Additionally, while the background initialization involves negligible overhead, the final streaming transfer bursts network traffic, and significantly reduces streaming efficiency if unexpected congestion occurs on the network. 

\paragraph{Concurrent reconfiguration events.}
The current design serializes reconfiguration events: if a new event arrives during the Switch/Transfer window, it is queued and processed after the current handoff completes. This design is motivated by the brevity of the critical window ($<$6\,s): the probability of a second event within this window is negligible under typical volatility patterns (events arrive on the order of tens of minutes). For the rarer case where the target topology becomes stale during preparation (e.g., a node in the new configuration is preempted before the Switch commits), the system can cancel the pending handoff and restart preparation with the updated topology, since the Active World's state remains intact until the Switch commits.

\paragraph{Preparation time vs.\ warning window.}
A practical concern is whether \sys's background preparation always fits within the preemption warning window. On our 32-GPU testbed, the total preparation time (Shadow World bootstrap + mock initialization + transfer planning) is under 60\,s, well within the typical 2-minute Spot notice. At 1,024 GPUs, our simulator projects this to grow to 90--150\,s, approaching but still within the window. For even larger scales, the scheduler can proactively trigger preparation before an official notice arrives---for example, when spot price trends indicate elevated preemption risk---thereby hiding the preparation latency entirely. The key property is that preparation is \emph{predictable} and can be overlapped with training, so it need not be triggered reactively.

\paragraph{Post-reconfiguration efficiency advantage.}
Existing online systems achieve fast reconfiguration but sacrifice flexibility in how parallelism may change. Bamboo~\cite{bamboo} and Oobleck~\cite{oobleck} restrict reconfiguration to adjustments within fixed pipeline templates---adding or removing replicas rather than reshaping tensor or pipeline parallelism degrees. Parcae~\cite{parcae} supports runtime DP+PP re-partitioning but excludes tensor parallelism; TrainMover~\cite{trainmover} preserves the original TP/PP/DP layout, performing only one-to-one node replacement. In contrast, \sys supports arbitrary mixed-parallel reshaping and scaling, and \sys reduces the downtime to second level in a scaling event.

\section{Conclusion}
\label{sec:conclusion}

We presented \sys, a live reconfiguration runtime that treats planned elasticity as a streaming network transfer rather than a storage operation.
Unlike traditional stop-and-restart systems, \sys replaces downtime with background initialization and layer-streamed transfer, executing arbitrary reshaping of multidimensional parallel configurations in seconds.
By focusing explicitly on environments with preemption warnings or planned scaling, \sys circumvents checkpointing bottlenecks and isolates initialization overheads. Evaluated on realistic large-scale workloads, \sys achieves 14$\times$--23$\times$ speedup compared to state-of-the-art checkpointing strategies, demonstrating its capability to maintain high training efficiency under volatile conditions.

%% file: appendix.tex
\section{Appendix}
\label{sec:appendix}

This appendix provides supplementary material: memory overhead analysis (\S\ref{sec:appendix-mem}), and intersection-based transfer formalization (\S\ref{sec:appendix-transfer}).

\subsection{Memory Overhead Analysis}
\label{sec:appendix-mem}

The Shadow World does not instantiate a second copy of model states. Table~\ref{tab:mem_overhead} provides a detailed per-rank breakdown of the additional memory consumed during a reconfiguration event.

\begin{table}[H]
\centering
\caption{Per-rank memory overhead during reconfiguration (GPT-14B, 32 GPUs, TP=4, PP=2, DP=4).}
\label{tab:mem_overhead}
\small
\begin{tabular}{lrr}
\toprule
\textbf{Component} & \textbf{Size} & \textbf{\% of 80\,GB} \\
\midrule
NCCL Communicator & $\sim$1\,GB& $\sim$1\%\\
Staging Buffer ($B$)          & 512\,MB      & 0.64\%    \\
Transfer Plan Metadata        & $\sim$1\,MB  & $<$0.01\% \\
\midrule
\textbf{Total}                & $\sim$1.5  GB& \textbf{$\sim$1.65\%} \\
\bottomrule
\end{tabular}
\end{table}

The total overhead is bounded by $O(B + C)$ where $B$ is the staging buffer size and $C$ is the communicator ---less than 1.5\,GB for a 14B model on 80\,GB A800 GPUs. For larger models, $B$ can be reduced, trading slightly longer transfer time for lower memory pressure. The staging buffer is reused across layers, so the overhead does not scale with model size.

\subsubsection{Formal Bounded-Memory Guarantee}

\begin{theorem}[Bounded Memory During Resharding]
During the execution of the Layer-Streaming Protocol (Algorithm~\ref{alg:streaming}), at any point in time, the peak additional memory consumption on any GPU rank is strictly bounded by $O(B + C)$, where $B$ is the staging buffer size and $C$ is the NCCL communicator metadata size.
\end{theorem}

\begin{proof}
Let $L$ be the number of layers and $S_\ell$ the total size of parameters and optimizer states for layer $\ell$. The protocol processes layers sequentially: at each step $\ell$, the additional memory on a destination rank consists of:
\begin{enumerate}
    \item The staging buffer $B$, allocated once before the loop (Algorithm~\ref{alg:streaming}, line~13) and reused across all layers.
    \item The assembled shard for layer $\ell$, which is written directly into the pre-allocated parameter storage of the new configuration (not additional memory, since this storage is required for training regardless).
\end{enumerate}
On a source rank, no additional memory is allocated: the source reads slices from existing parameter storage and sends them via \texttt{ISend} (line~9--10), which uses NCCL's internal send buffers (bounded by NCCL's own configuration, included in $C$).

The key invariant is \emph{no cross-layer accumulation}: the staging buffer is reused (line~13--14 re-allocate the same buffer each iteration, and Python/PyTorch's reference semantics ensure the previous layer's buffer is freed before the new one is allocated). The barrier at line~14 ensures that no rank proceeds to layer $\ell+1$ until all transfers for layer $\ell$ have completed, preventing any overlapping buffer usage.

Therefore, at any point during execution, the additional memory on any rank is at most $B + C$, independent of $L$ and $\sum_\ell S_\ell$. Since $B$ is a configurable constant (default 512\,MB--1\,GB) and $C$ is a small fixed overhead per communicator (typically $\sim$1\,GB as shown in Table~\ref{tab:mem_overhead}), the bound is $O(B + C)$.
\end{proof}

\subsection{Intersection-Based Transfer Planning}
\label{sec:appendix-transfer}

Standard frameworks tightly couple state tensor shards to rank indices. \sys introduces an \emph{Abstract Resource View}, where the training state is defined by logical tensors (e.g., ``Weights of Layer $L$'') and a sharding specification, rather than physical memory addresses. This abstraction allows us to treat reconfiguration as a geometric transformation problem.

\para{Design Goals.}
We design the resharding engine with three goals. First, it must be \emph{topology agnostic}, handling any transition ($TP \leftrightarrow TP'$, $DP \leftrightarrow DP'$, $PP \leftrightarrow PP'$) using a single unified algorithm. Second, the transfer must enforce \emph{bounded memory}, never materializing full tensors since model size far exceeds single-GPU memory. Third, all data must flow \emph{peer-to-peer} over the high-speed interconnect with zero disk usage.

\subsubsection{Problem Formalization}

Let $\mathcal{T} = \{T_1, T_2, \ldots, T_L\}$ be the set of tensors comprising the training state across $L$ layers, including parameters $\theta$, optimizer states $(m, v)$, and gradient accumulators (if needed).
We define a \emph{configuration} as a tuple $C = (R, \mathcal{S})$, where $R \subseteq \mathbb{N}$ is the set of participating GPU ranks and $\mathcal{S}: \mathcal{T} \times R \to \mathcal{P}(\mathbb{N}^d)$ is a sharding function that maps each tensor-rank pair to a set of tensor indices.

\sys supports three distinct classes of configuration transitions.
\textbf{In-place Reconfiguration} involves changing the sharding function $\mathcal{S}$ while the set of GPUs $R$ remains constant ($|R_{old}| = |R_{new}|$). This occurs, for example, when transitioning from $(TP=4, PP=2)$ to $(TP=8, PP=1)$ on a fixed 32-GPU cluster.
\textbf{Scale-Out} occurs when $R_{old} \subset R_{new}$, requiring state to be redistributed from dense shards to sparse shards to utilize added capacity.
Conversely, \textbf{Scale-In} ($R_{new} \subset R_{old}$) requires consolidating state onto fewer devices, typically to handle preemption or resource rebalancing.

We denote the shard of tensor $T$ owned by rank $r$ under configuration $C$ as $\mathcal{S}(T, r, C)$.
The fundamental correctness requirement is that the union of all shards in the new configuration must be equivalent to that in the old configuration:
\begin{equation}
    \forall T \in \mathcal{T}: \bigcup_{r \in R_{old}} \mathcal{S}(T, r, C_{old}) \equiv \bigcup_{r' \in R_{new}} \mathcal{S}(T, r', C_{new})
\end{equation}
The system must compute and execute the data movement plan to achieve this equivalence without materializing full tensors on any single GPU.

\subsubsection{Intersection-Based Transfer Planning}

We propose a \emph{distributed, geometry-aware planning algorithm} that computes the exact data transfers required for resharding.
The key insight is that each pair of (source rank, destination rank) can independently determine whether they need to exchange any data by computing the geometric intersection of their respective tensor views.

\begin{definition}[View Function]
For a tensor $T$ with shape $(d_1, d_2, \ldots, d_n)$, the \emph{view function} $\mathcal{V}(T, C, r)$ returns the $n$-dimensional hyper-rectangular region of indices owned by rank $r$ under configuration $C$. This region is represented as $\prod_{i=1}^{n} [l_i, h_i)$ where $l_i$ and $h_i$ are the lower and upper bounds along dimension $i$.
\end{definition}

Given configurations $C_{old}$ and $C_{new}$, each pair of ranks $(r_{src} \in R_{old}, r_{dst} \in R_{new})$ computes the intersection region $\mathcal{R}_{transfer}(T, r_{src}, r_{dst}) = \mathcal{V}(T, C_{old}, r_{src}) \cap \mathcal{V}(T, C_{new}, r_{dst})$.
If $\mathcal{R}_{transfer}$ is non-empty, the system generates a \texttt{TransferTask} specifying the precise slice bounds.
This computation utilizes only sharding metadata, avoiding any access to actual tensor data.

\subsubsection{Parallelism-Specific Instantiation}

This intersection-based formalism naturally generalizes across all forms of parallelism.
In Tensor Parallelism (TP), layers are typically sharded along input or output dimensions. Increasing TP from 4 to 8, for instance, causes each source rank to conceptually split its columns; the intersection computation identifies the exact sub-range to send to a new partner.
For Data Parallelism (DP), increasing the replica count necessitates a broadcast-like pattern where existing ranks send their full shards to new partners. The intersection logic automatically handles this, degenerating to the full source view.
Pipeline Parallelism (PP) transitions involve moving entire layers between stages; here, the intersection is either the complete tensor or simple null set.
For Expert Parallelism (EP), the view function extends to the expert dimension, enabling fine-grained migration of expert parameters.

The planning phase operates entirely on CPU and requires only $O(|\mathcal{T}| \cdot \max(|R_{old}|, |R_{new}|))$ intersection computations after pruning optimization.
For a 175B-parameter model with 96 layers and 1,024 ranks, the entire plan is generated in under 1 second, which is negligible compared to the data transfer time.

\subsubsection{Layer-Streaming Protocol Details}

Algorithm~\ref{alg:streaming} shows the concrete execution protocol used after transfer planning.
The key property is bounded memory: each destination rank reuses the same fixed staging buffer across layers.

\begin{algorithm}[t]
\caption{Layer-Streaming Resharding}
\label{alg:streaming}
\begin{algorithmic}[1]
\Require Source config $C_{old}$, destination config $C_{new}$, tensors $\mathcal{T}$
\Require Staging buffer size $B$
\State $\mathcal{P} \gets \textsc{ComputeTransferPlan}(C_{old}, C_{new}, \mathcal{T})$
\For{each layer $\ell = 1$ to $L$}
    \State $\mathcal{P}_{\ell} \gets \{p \in \mathcal{P} : p.layer = \ell\}$
    \For{each rank $r$ in parallel}
        \If{$r \in R_{old}$}
            \For{each task $p \in \mathcal{P}_{\ell}$ where $p.src = r$}
                \State $slice \gets T_{\ell}[p.bounds]$
                \State \textsc{ISend}($slice$, $p.dst$)
            \EndFor
        \EndIf
        \If{$r \in R_{new}$}
            \State Allocate staging buffer $buf$ of size $B$
            \For{each task $p \in \mathcal{P}_{\ell}$ where $p.dst = r$}
                \State \textsc{IRecv}($buf$, $p.src$)
            \EndFor
            \State \textsc{WaitAll}()
            \State Assemble received chunks into $T'_{\ell}$
        \EndIf
    \EndFor
    \State \textsc{Barrier}()
\EndFor
\end{algorithmic}
\end{algorithm}